\documentclass[journal]{IEEEtran}
\newcommand{\bx}{{\bf x}}
\newcommand{\bd}{{\bf d}}
\newcommand{\bbr}{{\bf r}}
\newcommand{\by}{{\bf y}}
\newcommand{\bby}{{\bf Y}}

\newcommand{\bbx}{{\bf X}}

 \usepackage{amsmath,amssymb,amsthm,color,graphicx,bm,mathtools}
\usepackage{algorithm,algorithmic}
\usepackage{cases}
\usepackage{caption}
\usepackage{subcaption}
\newcommand{\overbar}[1]{\mkern 1.5mu\overline{\mkern-1.5mu#1\mkern-1.5mu}\mkern 1.5mu}
\newcommand*\midpoint[1]{\overline{#1}}
\newtheorem{theorem}{Theorem}
\newtheorem{proposition}[theorem]{Proposition}
\begin{document}
\title{On The 2D Phase Retrieval Problem}

\author{\IEEEauthorblockN{Dani Kogan\IEEEauthorrefmark{1}, Yonina~C.~Eldar,~\IEEEmembership{Fellow,~IEEE}\IEEEauthorrefmark{2}, and Dan Oron\IEEEauthorrefmark{1}}
	\thanks{\IEEEauthorrefmark{1}D. Kogan and D. Oron are with the Department of Physics of Complex Systems,
		Weizmann Institute of Science, Rehovot 76100, Israel  ; e-mail: {\tt danikogan10@gmail.com}. \IEEEauthorrefmark{2}Y. C. Eldar is with the Department of Electrical Engineering, Technion--Israel Institute of Technology, Haifa 32000, Israel; e-mail: {\tt yonina@ee.technion.ac.il}.
		D. Oron acknowledges support from the Israeli Centers of Research Excellence programme and the Crown photonics center. Y. C. Eldar was supported in part by the European
		Union’s Horizon 2020 Research and Innovation Program through the
		ERCBNYQ Project, and in part by the Israel Science Foundation under Grant
		335/14.}}

\maketitle

\begin{abstract}
The recovery of a signal from the magnitude of its Fourier transform,
also known as phase retrieval, is of fundamental importance in many
scientific fields. It is well known that due to the loss of Fourier phase the problem in 1D is ill-posed. Without further constraints, there  is no unique solution to the problem. In contrast, uniqueness up to trivial ambiguities very often exists in higher dimensions, with mild constraints on the input. In this paper we focus on the 2D phase retrieval problem and provide insight into this uniqueness property by exploring the connection between the 2D and 1D formulations. In particular, we show that 2D phase retrieval can be cast as a 1D problem with additional constraints, which limit the solution space. We then prove that only one additional constraint is sufficient to reduce the many feasible solutions in the 1D setting to a unique solution for almost all signals. These results allow to obtain an analytical approach (with combinatorial complexity) to solve the 2D phase retrieval problem when it is unique.
\par
\end{abstract}


\section{Introduction}
Recovery of a signal from the modulus of its Fourier transform, also known as phase retrieval    \cite{patterson1934fourier}, \cite{patterson1944ambiguities}, is of paramount importance in many scientific fields such as optics \cite{walther1963question}, X-ray crystallography \cite{millane1990phase}, astronomy  \cite{fienup1987phase}, computational biology \cite{Stefik1978DNA}, speech recognition \cite{rabiner1993fundamentals} and more \cite{JEB16,shechtman2015phase}. Generally, the phase retrieval problem has no unique solution since any choice of the Fourier phase will generate a valid solution which can be far from the original signal \cite{hayes1982reconstruction}, \cite{bates1982fourier}. If the unknown input is compactly supported and the frequency domain is oversampled by a factor of two, then the phase retrieval problem can be equivalently stated as that of retrieving a vector from its autocorrelation function.

To study uniqueness in the 1D discrete case, Bruck and Sodin \cite{bruck1979ambiguity} considered the z-transform of the autocorrelation sequence. The fact that 1D polynomials can always be factored, leads to the conclusion that in general there is no uniqueness in 1D phase retrieval.  This is true even when we ignore trivial ambiguities
which include a global phase shift, conjugate inversion and spatial shift. Furthermore, there are generally as many as  $2^{N-1}$  vectors of length $N$ that share the same autocorrelation sequence \cite{H64}. An exception is when prior knowledge on the signal is available. One example is when the signal is minimum phase so that all the zeros of its z-transform are known to lie within the unit circle
\cite{HES06}.
For higher dimensions, and especially in 2D, Bruck and Sodin argued that since multivariate polynomials cannot in general be factored, phase retrieval has a unique solution for almost all signals.

  Despite the uniqueness guarantees in 2D, there is no general solution method available to find the unknown signal from its Fourier magnitude \cite{convex}.
      Over the years, several approaches have been suggested for solving the phase retrieval problem algorithmically.
       The most popular techniques are based on alternating projections, pioneered by Gerchberg and Saxton \cite{gerchberg1972practical} and extended by Fienup  \cite{fienup1978reconstruction,streibl1984phase,fienup1982phase}.  More recent approaches include semi-definite programming (SDP) algorithms \cite{shechtman2011sparsity,candes2015phase,waldspurger2015phase}, gradient techniques with appropriate initialization such as Wirtinger flow \cite{C2015}, and greedy methods with a sparsity prior \cite{BE12,shechtman2014gespar}.

In this work we study the 2D discrete phase retrieval problem with the goal of obtaining further insight into its uniqueness properties, and understanding the intrinsic differences and relationships between the 1D and 2D formulations. In particular, it is natural to attempt to describe the 2D autocorrelation matrix $\mathbf{R}$ of an $N \times N$ matrix $\bbx$ in terms of the 1D autocorrelation $\bbr$ of its vectorized version $\bx$ with length $N^2-1$. We will see that we can compute $\bbr$ from $\mathbf{R}$ but the reverse computation is not possible. Namely, the 2D autocorrelation $\mathbf{R}$ provides more knowledge on $\bbx$ than the 1D autocorrelation of its vectorized version. Nonetheless, we can formulate the 2D phase retrieval problem in terms of the 1D autocorrelation $\bbr$ and $(N-1)^2$ additional constraints which correspond to knowledge of certain autocorrelation values of $\bbx$. Next we show that to ensure uniqueness for almost all signals, it is sufficient to remove most of these constraints, and add a {\em single} constraint beyond knowing the autocorrelation sequence $\bbr$. More specifically, for almost all matrices $\bbx$, knowing the 1D autocorrelation $\bbr$ of its vectorized version together with the additional constraint $R(N-1,-(N-1))=x[N-1]x[N^2-N]$ is sufficient to reduce the many non-trivial solutions of the  1D problem and guarantee uniqueness. Finally, we illustrate how to use these results to construct the unique recovery when it exists.

This paper is organized as follows. In Section \ref{sec:pr}, we mathematically set up the phase retrieval problem and discuss uniqueness in one and two  dimensions. We then show how 2D phase retrieval can be recast as a 1D phase retrieval problem with additional constraints in Section~\ref{sec:reform}. In Section \ref{sec:unique} we prove that only one additional constraint from the reformulated problem is sufficient to guarantee uniqueness for almost all signals. We provide several examples to demonstrate the results.

\section{Phase Retrieval in One and Two Dimensions}
\label{sec:pr}
We begin by formulating the 1D and 2D phase retrieval problems, and discussing their basic uniqueness properties.

\subsection{1D Phase Retrieval}
\label{sec:pr1}

The 1D phase retrieval problem is to recover a vector $\bx$ from the magnitude of its Fourier transform. To set up the problem mathematically,
let $\mathbf{x}\in \mathbb{R}^N$ be a real length-$N$ vector. Its discrete Fourier transform (DFT)  of length $M \geq N$ is defined as
\begin{equation}
\widehat{x}[k]=\sum_{n=0}^{N-1} x[n] e^{-j \frac{2 \pi k}{M}},\quad 0 \leq k \leq M-1.
\end{equation}
 Since the DFT is a linear transformation it can be represented by a matrix-vector multiplication  $\mathbf{F}_{M,N}\mathbf{x}$, where $\mathbf{F}_{M,N}$ is the $M \times N$ matrix consisting of the first $N$ columns of the $M$-point DFT matrix
\begin{align}
\mathbf{F}_{M,N}& =\left(\begin{array}{ccccccc}
 1 & 1 & 1 &\dots& 1 \\
 1 & \phi & \phi^2 &\dots & \phi^{N-1} \\
\vdots & \vdots & \vdots &\vdots  & \vdots \\
 1 &\phi^{M-1} &\phi^{2(M-1)} & \dots& \phi^{(N-1)(M-1)} \\
\end{array}\right),
\end{align} and $\phi=e^{-j2\pi/M}$.
With these definitions we can state the 1D phase retrieval problem as
\begin{align}& \text{find}   \hspace{2cm} \mathbf{x}\nonumber \\
& \text{subject to } \hspace{1cm} \mathbf{y}=|\mathbf{F}_{M,N} \mathbf{x}|^2
\label{eq:algo_01}
\end{align}
 where $\by$ are the given measurements and the absolute value operation is taken element-wise.

Assuming $M \geq 2N-1$, phase retrieval can equivalently be formulated as the problem of reconstructing a signal from its autocorrelation measurements
\begin{equation}
r[\ell]=\sum_{i=0}^{N-1-\ell}x[i]x[i+\ell],\quad 0 \leq \ell \leq N-1,
\end{equation}
with $r[\ell]=r[-\ell]$.
For shorthand notation, we denote by $\bbr=\mathbf{x}\star \mathbf{x}$ the autocorrelation operator.
We can then write problem (\ref{eq:algo_01}) as
\begin{align}& \text{find} \hspace{2cm} \mathbf{x}\nonumber \\
& \text{subject to}  \hspace{1.2cm} \mathbf{r}=\mathbf{x}\star \mathbf{x}.
\label{eq:algo_02}
\end{align}

Due to the loss of phase, problem (\ref{eq:algo_01}) is ill posed and generally there are many possible solutions.
First, there are trivial ambiguities that result from the fact that the following three operations conserve the Fourier magnitude:
\begin{itemize}
\item [1.] Global phase shift $x[n]\rightarrow x[n]\cdot e^{ח \phi_0}$
\item [2.] Conjugate inversion $x[n]\rightarrow \overbar{x[-n]}$
\item [3.] Spatial shift $x[n]\rightarrow x[n+n_0]$.
\end{itemize}
Even up to trivial ambiguities, the 1D problem is not unique. In particular, any two sequences of length $N$ that lead to the same autocorrelation sequence will have the same Fourier magnitude.
In \cite{beinert2015ambiguities} the authors show that there are as many as $2^{N-2}$ signals that share the same autocorrelation and are not trivially related.
To see this, we analyze the roots of $\widehat{r}[k]$, the DFT of $r[\ell]$.

Let $P_\mathbf{r}(z)$ be the \emph{associated polynomial} to $\widehat{r}[k]$:
\begin{equation}
P_\mathbf{r}(z)\coloneqq \sum_{n=0}^{2N-2} r[n-N+1]z^n\label{eq:as_1}.
\end{equation}
Since $r[\ell]$ is symmetric, the zeros of $P_\mathbf{r}(z)$ appear in reflected pairs $(\gamma_i,\overbar{\gamma}_i^{-1})$ with respect to the unit circle and there are $N-1$ such zero pairs. The associated polynomial can therefore be factored as
\begin{equation}
P_\bbr(z)=r[N-1]\prod_{i=1}^{N-1}\left(z-\gamma_i\right)\left(z-\overbar{\gamma}_i^{-1}\right).
\label{eq:as_2}
\end{equation}
On the unit circle we have
\begin{equation}
\widehat{r}[k]=r[N-1]\prod_{i=1}^{N-1}\left(e^{-j2\pi k/M}-\gamma_i\right)\left(e^{-j2\pi k/M}-\overbar{\gamma}_i^{-1}\right).
\end{equation}
Recalling that $\widehat{r}[k]=|\widehat{x}[k]|^2$, it follows that given $\widehat{r}[k]$, we cannot determine whether a zero $\gamma_i$ or its conjugate reciprocal $\bar{\gamma}_i^{-1}$ is a root of $\widehat{x}[k]$ which leads to the non-uniqueness in phase retrieval.

 More specifically, each vector $\mathbf{x}\in \mathbb{C}^N$ with autocorrelation sequence $\mathbf{r}$ has an associated polynomial of the form (up to a phase factor):
 \begin{align}
P_{\mathbf{x}}(z)& = \sum_{i=0}^{N-1}x[i]z^{i}\nonumber \\
 &= \left[r[N-1]\prod_{j=1}^{N-1}\left|\beta_{j}\right|^{-1}\right]^{\frac{1}{2}}\prod_{j=1}^{N-1}(z-\beta_{j})\label{eq:as_3}
\end{align} where the parameters $\beta_{j}$ in (\ref{eq:as_3}) can be chosen
as $\beta_{j}\in\left\{ \gamma_{j},\bar{\gamma}_{j}^{-1}\right\} $. Aside from trivial ambiguities each solution is characterized by a different \emph{zero set} $B\coloneqq \left\{ \beta_1,...,\beta_{N-1}\right\}$. Obviously as $\beta_j$ is allowed to either be $\gamma_j$ or $\overbar{\gamma}^{-1}_j$ we can construct up to $2^{N-1}$ different  zero sets. It can be shown (see \cite[corollary 3.3]{beinert2015ambiguities})  that if a signal $\mathbf{x}$ of the form (\ref{eq:as_3}) has a zero set $\left\{ \beta_1,...,\beta_{N-1}\right\}$ then the conjugate reflected signal of $\mathbf{x}$ has a zero set $\left\{ \overbar{\beta}^{-1}_1,...,\overbar{\beta}^{-1}_{N-1}\right\}$. Thus we conclude that there are up to $2^{N-2}$  non trivially different vectors $\mathbf{x}$ with the same autocorrelation $\mathbf{r}$.

\subsection{2D Phase Retrieval}

We now turn to discuss the 2D phase retrieval problem.
In this case our unknown is an $N\times N$ real matrix $\mathbf{X}$. We assume we are given the magnitude-square of the 2D Fourier transform
$\bby=|\mathbf{F}_{M,N} \bbx \mathbf{F}_{M,N}^T|^2$.
Our problem can then be written as
\begin{align}& \text{find}   \hspace{1.5cm} \mathbf{X}\nonumber \\
& \text{subject to } \hspace{.5cm} \textrm{vec}\left(\mathbf{Y}\right)=|\left(\mathbf{F}_{M,N}\otimes\mathbf{F}_{M,N}\right) \textrm{vec}\left(\mathbf{X}\right)|^2,
\label{eq:algo_03}
\end{align}
where $\otimes$ denotes the Kronecker product and vec is the vectorization of a matrix. Assuming once again that
$M \geq 2N-1$, we can reformulate this problem in terms of the 2D autocorrelation of the matrix $\mathbf{X}$, where the 2D autocorrelation is defined by
\begin{eqnarray}
\label{eq:R2}
R(i,j)& = & \sum_{m=0}^{N-1-|i|}\sum_{n=0}^{N-1-|j|}X(m,n)X(m+i,n+j)
,\nonumber \\ &&\ \hspace*{1.2in} 0 \leq i,j \leq N-1, \nonumber \\
R(i,j)& = & \sum_{m=0}^{N-1-i}\sum_{n=-j}^{N-1}X(m,n)X(m+i,n+j)
,\nonumber \\ && \hspace*{.7in} i > 0, -(N-1) \leq j \leq -1,
\end{eqnarray}
and $R(-i,-j)=R(i,j)$. As in the 1D case, we use the notation $\mathbf{R}=\mathbf{X}\star \mathbf{X}$ to describe (\ref{eq:R2}). We can then write (\ref{eq:algo_03}) as
\begin{align}& \text{find}   \hspace{2cm} \mathbf{X}\nonumber \\
& \text{subject to } \hspace{1cm} \mathbf{R}=\mathbf{X}\star \mathbf{X}.
\label{eq:algo_04}
\end{align}

Similarly to the 1D setting, the ambiguities of the 2D phase retrieval problem depend on the factorization of a multivariate algebraic polynomial
into irreducible polynomials \cite{hayes1982reconstruction}. Since almost all 2D polynomials are irreducible, namely the reducible polynomials form a set of measure zero in the space of all polynomials \cite{hayes1982red}, almost every 2D signal can be uniquely recovered from
its 2D Fourier magnitude. Therefore, the uniqueness properties in 1D and 2D Fourier phase retrieval are markedly different.
Our goal in this paper is to obtain further insight into this difference from the point of view of the autocorrelation sequences rather than the traditional algebraic viewpoint.

\section{Reformulating the 2D Problem}
\label{sec:reform}

In order to understand the essential differences between the 2D and 1D phase retrieval problems, it is intuitive to try and write the 2D problem in terms of a vector $\bx$ representing the matrix $\bbx$. A natural approach is to vectorize $\bbx$ and then attempt to express the correlations in (\ref{eq:algo_04}) in terms of the vector $\bx$.
To be explicit, we will focus on the choice $\bx=\textrm{vec}\left(\mathbf{X^T}\right)$ which vectorizes $\bbx$ row-wise.
Two immediate questions arise in this context:
\begin{enumerate}
\item Can we express the 2D autocorrelation $\mathbf{R}=\mathbf{X}\star \mathbf{X}$ in terms of the 1D autocorrelation $\mathbf{r}=\mathbf{x}\star \mathbf{x}$?
\item Is the 1D autocorrelation $\mathbf{r}$ together with a small number of constraints sufficient to obtain uniqueness for almost all matrices $\bbx$?
\end{enumerate}

In this section we focus on the first question and show that $\mathbf{R}$ contains more information than $\bbr$. In other words, knowing $\bbr$ is not enough to recover $\mathbf{R}$. Beyond $\bbr$ we will need an additional $(N-1)^2$ constraints to fully characterize $\mathbf{R}$.
The second question will be discussed in Section~\ref{sec:unique} where we show that knowledge of $\bbr$ together with only one additional constraint that involves a specific value of $\mathbf{R}$ is sufficient to ensure uniqueness for almost all 2D signals.

\subsection{From 2D to 1D}

A first step towards reformulating the 2D problem into a  1D counterpart is to vectorize the matrix $\mathbf{X}$ so we effectively retrieve the same object as one would in a 1D setting. As noted above we will do that via $\mathbf{x}$=$\textrm{vec}\left(\mathbf{X^T}\right)$ which implies that
\begin{equation}
\label{eq:1d2}
x[mN+n]=X(m,n), \quad 0\leq m,n\leq N-1.
\end{equation}
The next step is to understand the relationship between the 2D autocorrelation $\mathbf{R}$ of $\bbx$ and the 1D autocorrelation $\bbr$ of $\bx$. This relationship is expressed in the following proposition.
\begin{proposition}
\label{prop:2d1d}
Let $\mathbf{X}$ be a real valued $N\times N$ matrix and let $\mathbf{x}=\textrm{vec}\left(\mathbf{X^T}\right)$ be its row vectorized representation of length $N^2$. Let $\mathbf{R}=\mathbf{X} \star \mathbf{X}$ be the 2D autocorrelation matrix of $\bbx$ and let $\bbr=\bx \star \bx$ be the
1D autocorrelation sequence of $\bx$. Then for values $0 \leq \ell \leq N^2-1$ we have
\begin{equation}
\label{eq:rviaR}
r(\ell)=\hspace*{-0.05in}\left\{\hspace*{-0.08in}
\begin{array}{ll}
R(i,0) & \hspace*{-0.02in} \ell=iN,  \\
R(N-1,j) & \hspace*{-0.02in} N(N-1)< \ell \leq N^2-1 \\
R(i,j)+R(i+1,j-N) & \hspace*{-0.02in} \mbox{otherwise},\\
\end{array}
\right.
\end{equation}
where $0 \leq i,j \leq N-1$ are the unique values satisfying $\ell=iN+j$, namely, $i=\lfloor \ell/N \rfloor$ and $j=\ell \mbox{ mod }N$. \end{proposition}
\begin{proof}
Substituting (\ref{eq:1d2}) into (\ref{eq:R2}) the 2D autocorrelation takes the form
\begin{eqnarray}
R(i,j)&=&\sum_{m=0}^{N-1-|i|}\sum_{n=0}^{N-1-|j|}x[mN+n]x[(m+i)N+n+j]\nonumber \\
&  = & \sum_{m=0}^{N-1-|i|}\sum_{n=0}^{N-1-|j|}x[mN+n]x[mN+n+\ell],
\label{eq:2d_2}
\end{eqnarray}
for $0 \leq i,j \leq N-1$ where $\ell=iN+j$. Similarly, for $i>0,-(N-1) \leq j \leq -1$ we have
\begin{equation}
\label{eq:2d_2n}
R(i,j) =  \sum_{m=0}^{N-1-|i|}\sum_{n=-j}^{N-1}x[mN+n]x[mN+n+\ell].
\end{equation}
Since the summand in $R(i,j)$ depends only on $\ell$, any valid choice of $i,j$ that leads to the same value of $\ell=iN+j$ will have the same summand. In particular, the pairs $(i,j)$ and $(i+1,j-N)$ result in the same $\ell$.
Therefore, for every $0<i \leq N-2$ and $0 < j \leq N-1$ there is a corresponding value $j'=j-N$ in the range $-(N-1) \leq j \leq -1$ with the same $\ell$. For this choice of $i,j'$  we have
\begin{equation}
R(i+1,j-N)  = \sum_{m=0}^{N-2-i}\sum_{n=N-j}^{N-1}x[mN+n]x[mN+n+\ell].
\end{equation}

Adding (\ref{eq:2d_2}) and (\ref{eq:2d_2n}) for $0<i \leq N-2,0 < j \leq N-1$ gives
\begin{eqnarray}
\lefteqn{R(i,j)+R(i+1,j-N)= } \nonumber \\  &\hspace*{-0.12in} = & \hspace*{-0.20in} \left(\sum_{m=0}^{N-2-i}\sum_{n=N-j}^{N-1}\hspace*{-0.08in} + \hspace*{-0.08in}\sum_{m=0}^{N-1-i}\sum_{n=0}^{N-1-j}\right)x[mN+n]x[mN+n+\ell]\nonumber\\
&\hspace*{-0.12in} = & \hspace*{-0.20in}  \sum_{m=0}^{N-2-i}\sum_{n=0}^{N-1}x[mN+n]x[mN+n+\ell]+\nonumber \\
&&  \hspace*{-0.12in} \sum_{n=0}^{N-1-j}\hspace*{-0.1in} x[(N-1-i)N+n]x[(N-1-i)N+n+\ell].
\label{eq:2d_4}
\end{eqnarray}
We next substitute $k=mN+n$ in the first term of (\ref{eq:2d_4}) and $s=(N-i-1)N+n$ in the second term which results in
\begin{eqnarray}
\lefteqn{R(i,j)+R(i+1,j-N)= } \nonumber \\
& = & \sum_{k=0}^{(N-i-1)N-1}x[k]x[k+\ell]+\hspace*{-0.2in} \sum_{s=(N-i-1)N}^{(N-i-1)N+N-1-j}x[s]x[s+\ell]\nonumber \\
&= & \sum_{k=0}^{N^2-1-(iN+j)}x[k]x[k+\ell]\nonumber \\
&= &\sum_{k=0}^{N^2-1-\ell}x[k]x[k+\ell].
\label{eq:2d_5}
\end{eqnarray}
The expression in (\ref{eq:2d_5}) is precisely the autocorrelation sequence of $\mathbf{x}$ for values of $\ell=iN+j$ with
  $0<i \leq N-2,0 < j \leq N-1$. This corresponds to all values $0 \leq \ell \leq N(N-1)-1$ excluding $\ell=Ni$.

It remains to determine $r(\ell)$ for $\ell=Ni$ and $N(N-1)< \ell \leq N^2-1$.
We begin with $\ell=Ni$ which corresponds to $j=0$. Substituting these values into (\ref{eq:2d_2}) results in
\begin{eqnarray}
R(i,0) & = & \sum_{m=0}^{N-1-|i|}\sum_{n=0}^{N-1}x[mN+n]x[mN+n+\ell]\nonumber \\
&  = & \sum_{k=0}^{(N-1-|i|)N+N-1}x[k]x[k+\ell] \nonumber \\
&  = & \sum_{k=0}^{N^2-1-\ell}x[k]x[k+\ell]=r[\ell],
\end{eqnarray}
where we defined $k=mN+n$. Next, the values $N(N-1)< \ell \leq N^2-1$ correspond to $i=N-1$ and $0 \leq j \leq N-1$. Plugging into (\ref{eq:2d_2}) leads to
\begin{eqnarray}
R(N-1,j) & = & \sum_{n=0}^{N-1-j}x[n]x[n+\ell]\nonumber \\
&  = & \sum_{k=0}^{N^2-1-\ell}x[k]x[k+\ell]=r[\ell],
\end{eqnarray}
since $\ell=(N-1)N+j$.
\end{proof}

Proposition~\ref{prop:2d1d} establishes that the 1D autocorrelation provides less knowledge on $\bbx$ than the 2D autocorrelation.
Indeed, knowing $\bbr$ provides information on $R(i,j)$ for $-(N-1) \leq j \leq -1$ only through the sum of these values with matching autocorrelation values.
Therefore, in order to be able to recover $\mathbf{R}$ from $\bbr$ we must obtain these values separately, namely we require an additional $(N-1)^2$ constraints on $\mathbf{R}$.
This observation combined with Proposition~\ref{prop:2d1d} leads to the following theorem, which establishes that the 2D phase retrieval problem can be cast as a 1D problem with additional constraints.
\begin{theorem}
\label{thm:2d1d}
Let $\mathbf{X}$ be a real valued $N\times N$ matrix and let $\mathbf{x}=\textrm{vec}\left(\mathbf{X^T}\right)$ be its row vectorized representation. Let $\mathbf{R}=\mathbf{X} \star \mathbf{X}$ be the 2D autocorrelation matrix of $\bbx$ and let $\bbr=\bx \star \bx$ be the
 1D autocorrelation sequence of $\bx$. Then the 2D phase retrieval problem
\begin{equation}
\begin{aligned}
&\mbox{find} \hspace{2cm} \bbx  \\
\nonumber & \mbox{subject to} \hspace{1.15cm}  \mathbf{R}=\bbx \star \bbx
\end{aligned}
\label{eq:algo_06}
\end{equation}
is equivalent to the 1D problem with $(N-1)^2$ additional constraints:
\begin{equation}
\begin{aligned}
&\mbox{find} \hspace{1cm} \bx  \\
\nonumber & \mbox{subject to} \hspace{.15cm}  \mathbf{r}=\mathbf{x} \star \mathbf{x}\\
 & \hspace{1.5cm}  R(i,j)= \sum_{m=0}^{N-1-i}\sum_{n=-j}^{N-1}\\
& \hspace{3.5cm} x[mN+n]x[mN+n+\ell]
\end{aligned}
\label{eq:algo_07}
\end{equation}
for $i>0,-(N-1) \leq j \leq -1$ where $\ell=iN+j$ and
$0\leq \ell\leq N^2-1$. The values of $r(\ell)$ are determined from $R(i,j)$ via (\ref{eq:rviaR}).
\end{theorem}

To summarize, the general procedure for reformulating a 2D problem as a 1D problem is as follows:
\begin{enumerate}
\item  Define the vector $\mathbf{x}$ that is generated by vectorizing the matrix $\mathbf{X}^T$, i.e $\mathbf{x}=\textrm{vec}\left(\mathbf{X}^T\right).$
\item  Extract the 1D autocorrelation sequence of $\mathbf{x}$ from the 2D autocorrelation matrix of $\mathbf{X}$ via Proposition~\ref{prop:2d1d}.
\item Using the extracted 1D autocorrelation, recast the 2D phase retrieval problem on $\mathbf{X}$ as a 1D phase retrieval problem on $\bx$ with additional $(N-1)^2$ constraints as stated in Theorem~\ref{thm:2d1d}.
\end{enumerate}

\subsection{An example}
We now provide an example demonstrating the steps above in order to reformulate a simple 2D problem.

Let $\mathbf{X}$ be the $2\times 2$ matrix given by
 \begin{align}
\mathbf{X} & =\left(\begin{array}{ccccccc}
  -24 & 26 \\
-9 & 1   \\
\end{array}\right).
\end{align}
In a phase retrieval experiment we measure the 2D autocorrelation of the matrix $\mathbf{X}$:
\begin{align}
\mathbf{R} & =\left(\begin{array}{ccccc}
 R(-1,-1) & R(-1,0) & R(-1,1) \\
 R(0,-1) & R(0,0) & R(0,1) \\
 R(1,-1) & R(1,0) & R(1,1) \\
\end{array}\right).
\end{align}
The autocorrelation matrix $\mathbf{R}$ has 5 unique values corresponding to the 5 constraints
\begin{align}
& X(0,0)^2+X(0,1)^2+X(1,0)^2+X(1,1)^2=R(0,0)\nonumber \\
&X(0,0)X(0,1)+X(1,0)X(1,1)=R(0,1)\nonumber\\
 &X(0,0)X(1,0)+X(0,1)X(1,1) =R(1,0)\nonumber\\
&X(0,0)X(1,1) =R(1,1)\nonumber\\
& X(0,1)X(1,0)=R(1,-1).
\end{align}
Computing these values results in
\begin{align}
\mathbf{R}
&= \left(\begin{array}{ccccc}
 -24 & 242 & -234  \\
 -633 & 1334 & -633 \\
 -234 & 242 & -24 \\
 \end{array}\right).
\end{align}
Therefore, the 2D phase retrieval problem can be stated as:
 \begin{align}
 \label{eq:ex2d}
& \mbox{find } \mathbf{X}  =\left(\begin{array}{ccccccc}
  X(0,0) & X(0,1) \\
X(1,0) & X(1,1)   \\
\end{array}\right) \nonumber \\
&  \mbox{subject to}  \nonumber \\
& \hspace*{.25in}X(0,0)^2+X(0,1)^2+X(1,0)^2+X(1,1)^2=1334 \nonumber \\
&\hspace*{.25in}X(0,0)X(0,1)+X(1,0)X(1,1)=-633\nonumber\\
 &\hspace*{.25in}X(0,0)X(1,0)+X(0,1)X(1,1) =242\nonumber\\
&\hspace*{.25in}X(0,0)X(1,1) =-24\nonumber\\
&\hspace*{.25in} X(0,1)X(1,0)=-234.
\end{align}

Next we show how to reformulate (\ref{eq:ex2d}) as a 1D problem with one additional constraint using the steps described at the end of the previous section:
\begin{enumerate}
\item   Define the vector  $\sloppy \mathbf{x}=\textrm{vec}\left(\mathbf{X}^T\right)=\left[X(0,0),X(0,1),X(1,0),X(1,1)\right]$ by vectorizing the matrix $\mathbf{X}^T$.
\item  Using Proposition~\ref{prop:2d1d} extract the 1D autocorrelation $\mathbf{r}$ of $\mathbf{x}$ which results in $\sloppy \mathbf{r}=\left[R(0,0),R(0,1)+R(1,-1),R(1,0),R(1,1)\right]=[1334,-867,242,-24]$ for the nonnegative part of $\bbr$.
\item  From Theorem~\ref{thm:2d1d} we can now recast the 2D phase retrieval problem  as
a 1D phase retrieval problem with
$(N-1)^2=1$ additional constraints:
\begin{align}& \text{\normalfont find}  ~\mathbf{x} \nonumber \\
& \text{\normalfont subject to}  \nonumber \\
& \hspace*{.25in} \bx \star \bx=[-24,242,-867,1334,-867,242,-24] \nonumber \\
 & \hspace*{.25in} x[1]x[2]=-234.
\label{eq:algo_08}
\end{align}
\end{enumerate}

\section{Ambiguity Reduction}
\label{sec:unique}
In Section \ref{sec:reform} we established how to recast 2D phase retrieval as a 1D problem with $(N-1)^2$ additional constraints. We now show that, for almost all signals, a single constraint is sufficient to reduce the many feasible solutions of the 1D formulation to one solution (up to trivial ambiguities).
In particular, we prove the following theorem.
\begin{theorem}
\label{thm:ambig}
Let $\mathbf{X}$ be a real valued $N\times N$ matrix and let $\mathbf{x}=\textrm{vec}\left(\mathbf{X}^T\right)$ be its row vectorized representation. Let $\mathbf{R}=\mathbf{X} \star \mathbf{X}$ be the 2D autocorrelation matrix of $\bbx$ and let $\bbr=\bx \star \bx$ be the
 1D autocorrelation sequence of $\bx$ defined by (\ref{eq:rviaR}).
Then the constraint $R(N-1,-(N-1))=x[N-1]x[N^2-N]$ alone is sufficient to reduce the non-trivial ambiguities of the 1D phase retrieval problem on $\mathbf{x}$. Specifically, the solution to
\begin{equation}
\begin{aligned}
&\mbox{find} \hspace{1cm} \bx  \\
\nonumber & \mbox{subject to} \hspace{.15cm}  \mathbf{r}=\mathbf{x} \star \mathbf{x}\\
 & \hspace{1.5cm}  R(N-1,-(N-1))=x[N-1]x[N^2-N]
\end{aligned}
\end{equation}
is unique for almost all vectors $\mathbf{x}$.
\end{theorem}

Theorem~\ref{thm:2d1d} implies that the problem of recovering a matrix $\bbx$ from the magnitude of its 2D Fourier transform can be recast as a  1D phase retrieval problem on the vector $\textrm{vec}\left(\mathbf{X}^T\right)$ with additional constraints. Theorem~\ref{thm:ambig} proves that it is sufficient to choose the constraint $R(N-1,1-N)$ to resolve the non trivial ambiguities in the 1D problem on $\textrm{vec}\left(\mathbf{X}^T\right)$. These insights can be used to retrieve a 2D signal from its Fourier magnitude using the following four steps:
\begin{enumerate}
\item Extract the 1D autocorrelation $\bbr=\bx\star \bx$ of $\bx=\textrm{vec}\left(\mathbf{X}^T\right)$ from the 2D autocorrelation $\mathbf{R}$ of the matrix $\mathbf{X}$ using (\ref{eq:rviaR}).
\item  Construct the associated polynomial (\ref{eq:as_1}) of $\mathbf{r}$.  Using the factorization (\ref{eq:as_2}) and (\ref{eq:as_3}) find all non-trivially different vectors with autocorrelation $\mathbf{r}$.
\item Search for the vector that fulfills the constraint $R(N-1,1-N)=x[N-1]x[N^2-n]$. Theorem~\ref{thm:ambig} suggests that almost always there will be exactly one such vector.
\item Reshape the vector back into a matrix. This matrix is the unique solution to the 2D phase retrieval problem considered.
\end{enumerate}
Note that the steps above have combinatorial complexity and therefore are practical primarily for small system sizes. However, they provide an analytic approach to solving 2D phase retrieval which can be of interest.

\subsection{Proof of Theorem~\ref{thm:ambig}}

Our derivations below rely on the results of \cite{beinert2015ambiguities} and \cite{beinert2016enforcing}.

In Theorem~\ref{thm:2d1d} we showed that the 2D phase retrieval problem with respect to $\bbx$ is equivalent to 1D phase retrieval of $\mathbf{x}$ with $(N-1)^2$ additional constraints. One of these constraints is on the value of  $x[N-1]x[N^2-N]$ which must be equal to $R(N-1,1-N)$. This constraint follows from substituting $i=N-1,j=1-N$ and $\ell=(N-1)^2$ into the expression for $R(i,j)$ in the theorem. For this choice, $m=0$ and $n=N-1$ so that the sum reduces to the single element $x[N-1]x[N^2-N]$.
We now show that, for almost all signals, there is only one trivially non-different vector that fulfills this constraint.

We begin by considering the 1D problem $\bbr=\bx \star \bx$ where $\bx$ has length $N^2-1$. We define the $\emph{autocorrelation vector}$ of $\mathbf{r}$ as the vector of zeros  $\bm{\gamma}=\left[\gamma_1,...,\gamma_{N^2-1}\right]$. Any autocorrelation sequence can be uniquely defined by its autocorrelation vector up to scale. In Section~\ref{sec:pr1} we have seen that knowing $\bbr$ is equivalent to knowing the polynomial $P_\bbr(z)$ of (\ref{eq:as_2}) where now there are $N^2-1$ pairs of zeros.
  In addition, we assume that there are $2^{N^2-2}$ distinct possible ways to factor this polynomial as in (\ref{eq:as_3}) such that the corresponding vectors are not trivially related. Denote by $A=\left\{\mathbf{y}_1,...,\mathbf{y}_{L}\right\}$ with $L=2^{N^2-2}$ the set of all real vectors $\mathbf{y}_i\in\mathbb{R}^{N^2}$ with autocorrelation $\mathbf{r}=\mathbf{x}\star\mathbf{x}$ that are non-trivially different, and let $\mathbf{y}\in A$ be a vector with zero set $B=\left\{\beta_1,...,\beta_{N^2-1}\right\}$ and associated polynomial
\begin{align}
&\hspace*{-0.2in} P_{\mathbf{y}}(z)=y[0]+y[1]z+...+y[N^2-1]z^{N^2-1} \nonumber \\
&=\left[r[N^2-1]\prod_{j=1}^{N^2-1}\left|\beta_{j}\right|^{-1}\right]^{\frac{1}{2}}\prod_{j=1}^{N^2-1}(z-\beta_{j}).
\label{eq:1}
\end{align}

We next rely on Vieta's formula \cite{hazewinkel2001viete} which relates a polynomial's coefficients to sums and products of its roots. In what follows we assume that all terms are given up to a global phase factor. Vieta's result then implies that
\begin{equation}
\frac{y[N^2-1-k]}{y[N^2-1]}=\left(\sum_{1\leq i_1\leq i_2 \leq ...\leq i_k \leq N^2-1}\beta_{i_1}\beta_{i_2}...\beta_{i_k}\right).
\end{equation}
Choosing $k=N-1$ leads to
\begin{equation}
y[N^2-N]=y[N^2-1]\left(\sum_{\tilde{N}^{2}}\beta_{i_1}\beta_{i_2}...\beta_{i_{N-1}}\right), \label{eq:4}
\end{equation}where $\sum_{\tilde{N}^{2}}=\sum_{1\leq i_{1}<i_{2}<...<i_{N-1}\leq N^{2}-1}$.
Let $\mathbf{y}^{(I)}$ be the vector with associated polynomial

\begin{eqnarray}
P_{\mathbf{y}^{(I)}}(z)&=&\left[r[N^2-1]\prod_{j=1}^{N^2-1}\left|\beta_{j}\right|\right]^{\frac{1}{2}}\prod_{j=1}^{N^2-1}(z-\midpoint{\beta}^{-1}_{j})\nonumber\\ \label{eq:3}
&=& y[N^2-1]+...+y[0]z^{N^2-1}.
\end{eqnarray}
The vector $\mathbf{y}^{(I)}$ is the flipped version of the vector $\mathbf{y}$ where $y^{(I)}[i]=y[N^2-1-i]$. Using Vieta's formula we have for $y^{(I)}[N^2-N]=y[N-1]$,

\begin{equation}
y^{(I)}[N^{2}-N]\ =y^{(I)}[N^2-1]\sum_{\tilde{N}^{2}}\overbar{\beta}^{-1}_{i_{1}}...\overbar{\beta}^{-1}_{i_{N-1}}.\label{eq:5}
\end{equation} The product of (\ref{eq:4}) and (\ref{eq:5}) is given by:

\begin{align}
&y[N^2-N]y[N-1]=y[N^2-N]y^{(I)}[N^2-N]\nonumber\\
&= \left|[r[N^{2}-1]\right|\left(\sum_{\tilde{N}^{2}}\beta_{i_{1}}...\beta_{i_{N-1}}\right)\left(\sum_{\tilde{N}^{2}}\overbar{\beta}^{-1}_{i_{1}}...\overbar{\beta}^{-1}_{i_{N-1}}\right),
\end{align} where we used the fact that from \eqref{eq:1} and \eqref{eq:3},

\begin{align}
&y[N^2-1]y^{(I)}[N^2-1] =\nonumber\\
&=\left[r[N^2-1]\prod_{j=1}^{N^2-1}\left|\beta_{j}\right|\right]^{\frac{1}{2}}\left[r[N^2-1]\prod_{j=1}^{N^2-1}\left|\beta_{j}\right|^{-1}\right]^{\frac{1}{2}}\nonumber\\
&=|r[N^2-1]|.
\end{align}

Since $\beta_i\in\left\{\gamma_i,\midpoint{\gamma}^{-1}_i\right\}$, the value $y[N^2-N]y[N-1]$ defines a function $f_\mathbf{y}(\bm{\gamma})=y[N-1]y[N^2-N]$ in the autocorrelation vector $\bm{\gamma}$. Assuming $\mathbf{y}$ is real, any complex roots come in conjugate pairs. Suppose, without lose of generality, that $\mathbf{y}_1=\mathbf{x}$ and that the zero set of $\mathbf{y}_1$ is given by $B=\left\{\gamma_1,...,\gamma_{N^2-1}\right\}$. The function $f_{\mathbf{y}_1}(\bm{\gamma})$ then becomes

\begin{eqnarray}
\lefteqn{\hspace*{-0.15in} f_{\mathbf{y}_1}(\bm{\gamma})=y_1[N-1]y_1[N^2-N]}\nonumber \\
&\hspace*{-0.2in} =&\hspace*{-0.15in} \left|[r[N^{2}-1]\right|\hspace*{-0.05in}\left(\sum_{\tilde{N}^{2}}\gamma_{i_{1}}...\gamma_{i_{N-1}}\right)\hspace*{-0.06in}\left(\sum_{\tilde{N}^{2}}\overbar{\gamma}^{-1}_{i_{1}}...\overbar{\gamma}^{-1}_{i_{N-1}}\right)
\hspace*{-0.05in}.
\end{eqnarray}

Suppose that we are given a specific autocorrelation sequence $\mathbf{\gamma_0}$ for which  the signal $\mathbf{x}$ cannot be uniquely recovered up to trivial ambiguities with the additional constraint $R(N-1,-(N-1))$ alone. Then there exists a second solution of the form $\mathbf{y}\in A\setminus{\{\mathbf{y}_1\}}$
with a zero set defined by $B^{\left({\Gamma}\right)}$, where
\begin{numcases}{B^{\left({\Gamma}\right)}=}
  \overbar{\gamma}_j^{-1}, & for $ j\in \Gamma$ \nonumber\\
 \gamma_j, & otherwise\label{eq:apx1}
\end{numcases}
and $\Gamma$ is a nonzero subset of the index set $S=\left\{1,...,N^2-1\right\}$. For this choice,

\begin{equation}
F_{{\mathbf{y}_1},{\mathbf{y}_0}}(\bm{\gamma}_0)=f_{\mathbf{y}_1}(\bm{\gamma}_0)-f_{\mathbf{y}}(\bm{\gamma}_0)=0.
\end{equation}
This implies that the number of functions $F_{{\mathbf{y}_1},{\mathbf{y}}_i}(\bm{\gamma}_0)$, $1< i \leq L$ that are zero for this given autocorrelation corresponds to the number of non-trivially different solutions with the same value $R(N-1,1-N)$.

Let us now define $\Lambda$ as the set of all autocorrelations for which the function $F_{{\mathbf{y}_1},{\mathbf{y}}}(\bm{\gamma})$ is zero, i.e

\begin{equation}
\Lambda=\left\{\bm{\gamma}\in D~\biggr{|}F_{{\mathbf{y}_1,\mathbf{y}}}(\bm{\gamma})=0\right\},\label{eq:apx2}
\end{equation} where $D$ is the domain on which the function is defined (see the Appendix for more details). The following  theorem, whose proof is provided in the Appendix, shows that for almost all autocorrelations the function is not zero.

\begin{theorem}
\label{thm:zero}
The zero set $\Lambda$ of $F_{\mathbf{y}_1,\mathbf{y}}(\bm{\gamma})$ has measure zero.
\end{theorem}
From Theorem~\ref{thm:zero} we conclude that the probability that for a given autocorrelation (a given $\bm{\gamma}\in D$) there is at least one function $F_{{\mathbf{y}_1},{\mathbf{y}_i}}(\bm{\gamma})=0$ , $ 1< i \leq L$ is zero. This shows that for almost all signals the constraint $R(N-1,-(N-1))$ alone is sufficient to guarantee uniqueness.

Note that in our derivation we considered the specific constraint $R(N-1,-(N-1))$. We conjecture that our conclusions will hold true for other choices of the $(N-1)^2$ constraints; however, we do not pursue this here.

\subsection{An Example}

In the example in Section~\ref{sec:reform} we considered a simple $2\times2$ matrix and showed how to recast the corresponding 2D phase retrieval problem into 1D phase retrieval. The reformulated 1D problem we obtained  was:
\begin{align}& \text{\normalfont find}  ~\mathbf{x} \nonumber \\
& \text{\normalfont subject to}  \nonumber \\
& \hspace*{.25in} \bx \star \bx=[-24,242,-867,1334,-867,242,-24] \nonumber \\
 & \hspace*{.25in} x[1]x[2]=-234.
 \label{eq:algo_09}
\end{align}

The solution to this problem without the additional constraint  $x[1]x[2]=-234$ can be found by considering the associated polynomial of the autocorrelation sequence ($\ref{eq:as_1}$) and its factorization ($\ref{eq:as_2}$):
\begin{align}
&P_\bbr(z)=\sum_{n=0}^{6}r[n-N^2+1]z^n \nonumber \\
&=-24+242z-633z^2+1334z^3-633z^4+242z^5-24z^6 \nonumber \\
&=-24(z-2)(z-3)(z-4)(z-1/2)(z-1/3)(z-1/4).
\end{align}
The autocorrelation vector of $\mathbf{r}$ is given by $\bm{\gamma}_0=\left[2,3,4\right]$.

There are $2^{4-1}=8$ different vectors with autocorrelation $\mathbf{r}$ (up to global phase and a shift) which can be found by considering all associated polynomials with zero sets $B=\{\beta_1,\beta_2,\beta_3\}$ where $\beta_1\in\{2,1/2\},~\beta_2\in\{3,1/3\},~\beta_3\in\{4,1/4\}$. Using ($\ref{eq:as_3}$), the associated polynomials are given by:
\begin{align}
&P_{\mathbf{y}_1}(z)=-24+26z^2-9z+z^3 \nonumber \\
&P_{\mathbf{y}_2}(z)=-6+29z-21z^2+4z^3\nonumber\\
&P_{\mathbf{y}_3}(z)=-8+30z-19z^2+3z^3\nonumber\\
&P_{\mathbf{y}_4}(z)=-2+15z-31z^2+12z^3\nonumber\\
&P_{\mathbf{y}_5}(z)=-12+31z-15z^2+2z^3\nonumber\\
&P_{\mathbf{y}_6}(z)=-3+19z-30z^2+8z^3\nonumber\\
&P_{\mathbf{y}_7}(z)=-4+21z-29z^2+6z^3\nonumber\\
&P_{\mathbf{y}_8}(z)=-1+9z-26z^2+24z^3.
\end{align}
 Solutions $\by_{8,...,5}$ are just a flipped version of  $\by_{1,...,4}$ (times -1) and are therefore  trivially different. We conclude that the phase retrieval problem  has $2^{4-2}=4$ non-trivially different solutions given by
\begin{align}
\left(\begin{array}{c}
\mathbf{y}_{1}\\
\mathbf{y}_{2}\\
\mathbf{y}_{3}\\
\mathbf{y}_{4}\\
\end{array}\right)=\left(\begin{array}{cccc}
-24 & 26 & -9 & 1\\
-6 & 29 & -21 & 4\\
-8 & 30 & -19 & 3\\
-2 & 15 & -31 & 12\\
\end{array}\right).
\end{align}
The additional constraint in (\ref{eq:algo_09}) states that $x[1]x[2]=-234$. It is easy to see that only the vector $\mathbf{y}_{1}$ satisfies this constraint. Therefore, this vector uniquely solves (\ref{eq:algo_09}). By reshaping $\mathbf{y}_{1}$ into its matrix form we conclude that the unique solution to (\ref{eq:algo_09}) is
 \begin{align}
\mathbf{X} & =\left(\begin{array}{ccccccc}
  -24 & 26 \\
-9 & 1   \\
\end{array}\right).
\end{align}

In the example above the 1D phase retrieval reformulation has only one additional constraint to begin with due to the small problem size. To show the ambiguity reduction for a problem with higher dimension we consider the case in which $N=3$. The equivalent 1D problem has dimension
 $N^2=9$. Let $\mathbf{x}_0$ be a randomly generated vector with the autocorrelation sequence $\bbr_0$. We have $2^{N^2-2}=128$ non trivially different vectors $\mathbf{y}_i,~0\leq i\leq 127$ that share the autocorrelation sequence $\mathbf{r_0}$. For $N=3$ the additional constraint in Theorem 2 is given by $x[N^2-N]x[N-1]=x[6]x[2]$. In Fig.~\ref{fig:ar1} we plot the vector $\bd$ with elements $d[i]=c_1 y_i[6]y_i[2]$ arranged in ascending order where $c_1$ is defined in such a way that $d[127]=1$. In Fig.~\ref{fig:ar2} we plot the vector $\mathbf{v}$ where $v[i]=\textrm{log}\left(d[i+1]-d[i]\right)$ for $0\leq i\leq 126$.
  Clearly the entries of $\bd$ have distinct values. Therefore, prior knowledge on the value of $y_i[6]y_i[2]$ indeed guarantees uniqueness.

\begin{figure}[!tbp]
	\centering
	\begin{minipage}[b]{0.4\textwidth}
		\includegraphics[width=\textwidth]{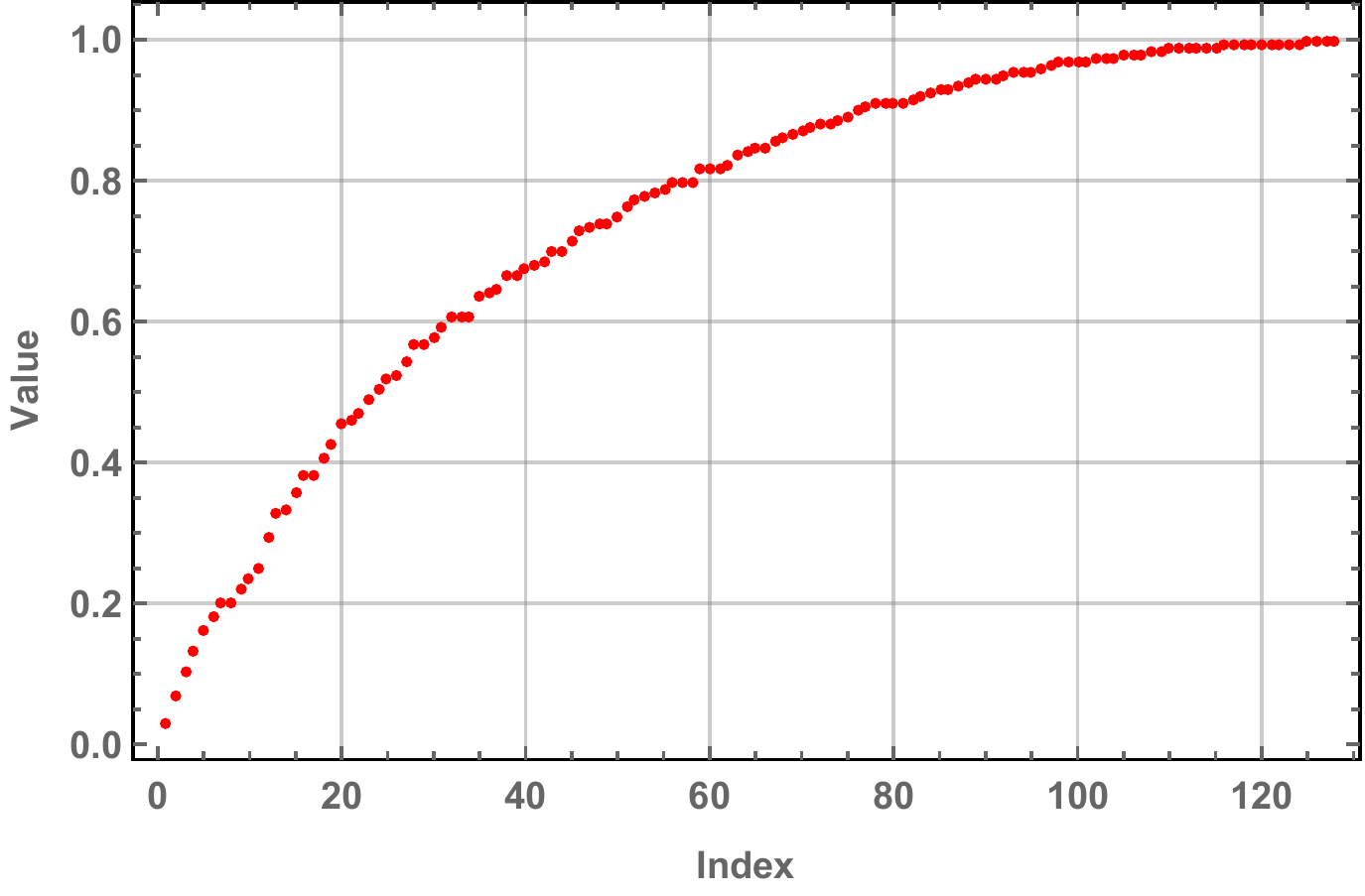}
		\caption{The values $c_1y_i[6]y_i[2]$ of all 128 non-trivially different vectors with the  autocorrelation sequence $\mathbf{r}_0$.}\label{fig:ar1}
	\end{minipage}
	\hfill
	\begin{minipage}[b]{0.4\textwidth}
		\includegraphics[width=\textwidth]{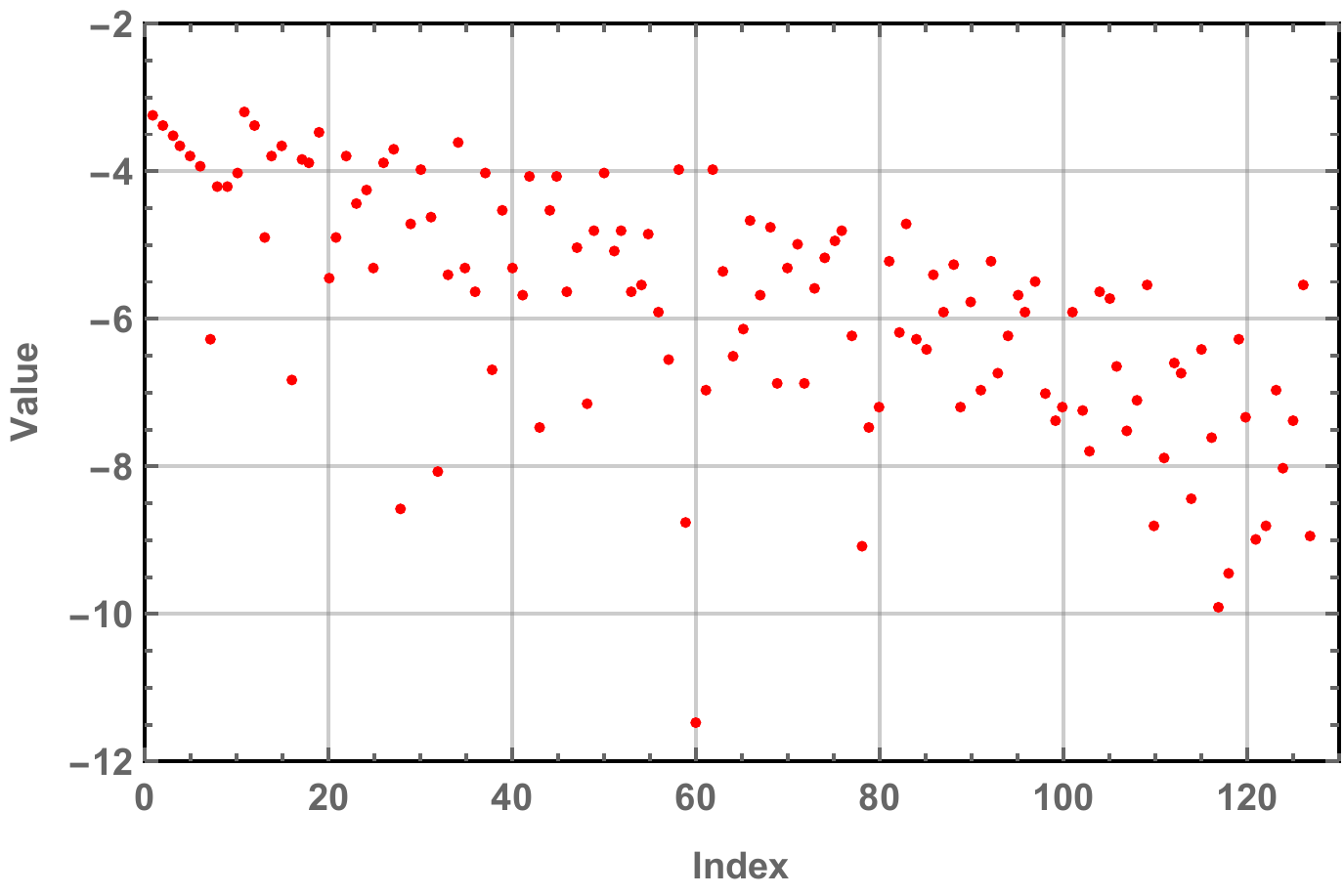}
		\caption{The vector $\mathbf{v}$. Since all values are finite the additional constraint $R(2,-2)=x[6]x[2]$ guarantees uniqueness.}\label{fig:ar2}
	\end{minipage}
\end{figure}

\section{Conclusion}
In this work we considered 2D phase retrieval and showed that it can be restated as a 1D phase retrieval problem with additional constraints. We then proved that one additional constraint is sufficient to reduce the many feasible solutions of the complementary 1D phase retrieval problem to a single solution for almost all 2D cases. This approach can be used to obtain an explicit solution for almost all 2D phase retrieval problems.

\appendix

In this appendix we prove Theorem~\ref{thm:zero}. In particular, we show that the zero set of $F_{{\mathbf{y}_1},{\mathbf{y}}}(\bm{\gamma})$ has measure zero.

As a first step, we prove the following proposition.

\begin{proposition}
\label{prop:zero}
The function $F_{{\mathbf{y}_1},{\mathbf{y}}}(\bm{\gamma})=f_{\mathbf{y}_1}(\bm{\gamma})-f_{\mathbf{y}}(\bm{\gamma})$ is not the zero function.
\end{proposition}
\begin{proof}

 Let $\mathbf{y}$ be the vector set $B^{(\Gamma)}$ defined in (\ref{eq:apx1}) and $S=\left\{1,...,N^2-1\right\}$. To show that  $F_{{\mathbf{y}_1},{\mathbf{y}}}(\bm{\gamma})$ is not the zero
function it is enough to show that it is not zero at a single point. Select one index $i_1\in S \setminus {\Gamma}$ and one index $i_2 \in \Gamma$,  and consider the autocorrelation defined by $
\tilde{\gamma}_{i_1}=\alpha,\,\tilde{\gamma}_{i_2}^{-1}=\alpha$
and $\tilde{\gamma}_{i}=1$ for $i\in S\setminus{\{i_1,i_2\}}$ where we choose $\alpha\gg1$.
Using simple combinatorics it can be shown that up to a phase factor

\begin{equation}
f_{\mathbf{y}_{1}}(\tilde{\bm{\gamma}})=\alpha^{2}\left[\left(k_{2}-k_{3}\right)^{2}+\mathcal{O}(\alpha^{-1})\right],
\end{equation}
and
\begin{equation}
f_{\mathbf{y}}(\tilde{\bm{\gamma}})   =\alpha^{2}\left\{k_{3}\left[k_{1}-2\left(k_{2}-\frac{1}{2}k_{3}\right)\right]+\mathcal{O}(\alpha^{-1})\right\},
\end{equation}where $k_{i}=\left(\begin{array}{c}
N^{2}-i\\
N-i
\end{array}\right)$. Then for $\alpha \gg 1$ we have

\begin{equation}
f_{\mathbf{y}_{1}}(\tilde{\bm{\gamma}})-f_{\mathbf{y}}(\tilde{\bm{\gamma}})=\left(k_2^2-k_3k_1\right)\alpha^2 \label{eq:apdx_1}.
\end{equation} Using the fact that $k_{i}=\frac{N^{2}-i}{N-i}k_{i+1}$ leads to

\begin{equation}
k_{1}k_{3}  =\left(\frac{N^{2}-1}{N-1}\right)\left(\frac{N-2}{N^{2}-2}\right)k_{2}^{2},
\end{equation} 
and (\ref{eq:apdx_1}) becomes
\begin{equation}
f_{\mathbf{y}_{1}}(\tilde{\bm{\gamma}})-f_{\mathbf{y}}(\tilde{\bm{\gamma}})=k_2^2\left[1-\left(\frac{N^{2}-1}{N-1}\right)\left(\frac{N-2}{N^{2}-2}\right)\right]\alpha^2.
\end{equation} 
 Since $\left(\frac{N^{2}-1}{N-1}\right)\left(\frac{N-2}{N^{2}-2}\right) \ne 1$ the above expression is not zero and therefore the function $F_{{\mathbf{y}_1},{\mathbf{y}}}(\bm{\gamma})$ is not identically zero.
\end{proof}

We now analyze the function $f_{\mathbf{y}}(\bm{\gamma})$.
 The domain and the image of this function are defined such that $f_{\mathbf{y}}: \mathbb{C}^{N^2-1}\setminus {G}\rightarrow\mathbb{C}$ where $G$ is the set for which $\gamma_i=0,~i\in \{1,...,N^2-1\}$. If we instead define $\bm{\gamma}=\left[\Re{\gamma_1},\Im{\gamma_1},...,\Re{\gamma_{N^2-1}},\Im{\gamma_{N^2-1}}\right]$ then $f_{\mathbf{y}}: \mathbb{R}^{2(N^2-1)}\setminus {G}\rightarrow\mathbb{C}$.

We restrict ourselves to real vectors so that any complex roots $\gamma_i$  appear in conjugate pairs. Thus, we may define a set of functions, that represent the function $f_\mathbf{y}(\bm{\gamma})$, that are defined as $f_{\mathbf{y},n}: \mathbb{R}^{N^2-1}\setminus {G}\rightarrow\mathbb{R}$ with $0\leq n\leq \lfloor{(N^2-1)/2}\rfloor$ where $n$ denotes the number of pairs of the roots that are not real valued. To understand how, we provide a simple example. Consider the case $N=2$. The function $f_{\mathbf{y}_1}(\bm{\gamma})$ is then given by (up to a global phase factor):
\begin{equation}
f_{\mathbf{y}_1}(\bm{\gamma})=r[3]\left(\gamma_1+\gamma_2+\gamma_3\right)\left(\overbar{\gamma}_1^{-1}+\overbar{\gamma}_2^{-1}+\overbar{\gamma}_3^{-1}\right).
\end{equation} In this case the number of pairs of complex roots can either be $n=0$ or $n=1$. If $n=1$ then we must have $\gamma_2=\overbar{\gamma}_1$. Expressing each root in terms of its real and complex part we have $\gamma_j=a_j+i b_j,~1\leq j\leq 3$ and the vector of variables becomes $\mathbf{d}=[a_1,b_1,a_3]$. The function takes the form
\begin{equation}
f_{\mathbf{y}_1,1}(\mathbf{d})=r[3]\frac{2 a_1 b_1^2+a_3 b_1^2+2 a_1^3+5 a_3 a_1^2+2 a_3^2 a_1}{a_3 \left(a_1^2+b_1^2\right)}.
\end{equation} On the other hand if all the roots are real we have $b_j=0,~1\leq j\leq 3$ and the requirement that $\gamma_2=\overbar{\gamma}_1$ no longer needs to be fulfilled. The vector $\bm{\gamma}$ becomes $\mathbf{d}=[a_1,a_2,a_3]$ and the function takes the form
\begin{equation}
f_{\mathbf{y}_1,0}(\mathbf{d})=r[3]\frac{\left(a_1+a_2+a_3\right) \left(a_2 a_3+a_1 \left(a_2+a_3\right)\right)}{a_1 a_2 a_3}.
\end{equation} 
 For this case the function $f_{\mathbf{y}_1}(\bm{\gamma})$ which is defined on $\mathbb{C}^{6}\setminus G$ can be fully represented by the two functions $f_{\mathbf{y}_1,0}(\mathbf{d})$ and $f_{\mathbf{y}_1,1}(\mathbf{d})$ that are defined on $\mathbb{R}^3\setminus{G}$.

For our problem the function $F^{(n)}_{\mathbf{y1},\mathbf{y}}(\mathbf{d})\coloneqq f_{\mathbf{y}_1,n}(\mathbf{d})-f_{\mathbf{y},n}(\mathbf{d})$ where $\mathbf{d}\in \mathbb{R}^{N^2-1}$ can be written as
\begin{equation}
F^{(n)}_{{\mathbf{y}_1,\mathbf{y}}}(\mathbf{d})=\frac{g_1(\mathbf{d})}{g_2(\mathbf{d})}
\end{equation} where $g_1(\mathbf{d})$ and $g_2(\mathbf{d})$ are some multivariate polynomials with real coefficients in $\mathbf{d}$.
The zero set $V_n$ of $F^{(n)}_{{\mathbf{y}_1,\mathbf{y}}}(\mathbf{d})$ is defined as $V_n=\left\{\mathbf{d}\in \mathbb{R}^{N^2-1}\setminus {G}~\biggr{|}F^{(n)}_{{\mathbf{y}_1,\mathbf{y}}}(\mathbf{d})=0\right\}$.
The zero set of $g_1(\mathbf{d})$ is given by $W=\left\{\mathbf{v}\in \mathbb{R}^{N^2-1}~\biggr{|} g_{1}(\mathbf{v})=0\right\}$. Since $g_1(\mathbf{d})$ is a real analytic function, which is not the zero function (as shown in Proposition~\ref{prop:zero}), and is defined on an open and connected set, the equation $g_1(\mathbf{d})=0$ defines a hyperplane on the real space $\mathbb{R}^{N^2-1}$ and its zero set is therefore of measure zero. For a rigorous proof the reader is referred to \cite[corollary 10, p.9]{gunning2009analytic}  for complex valued analytical functions and \cite{2015arXiv151207276M} for real valued analytical functions. The function $F^{(n)}_{{\mathbf{y}_1,\mathbf{y}}}(\mathbf{d})$ is zero if and only if $g_1(\mathbf{d})$ is zero so that $V_n\subseteq W$. Since $W$ is a set of measure zero, so is $V_n$. The zero set $\Lambda$ of the function $F_{\mathbf{y}_1,\mathbf{y}}(\bm{\gamma})$ is given by the union of all zero sets $V_n$,
\begin{equation}
\Lambda={\bigcup_{i= 1}^{\lfloor{(N^2-1)/2}\rfloor}}V_i.
\end{equation} A countable union of sets of measure zero is still measure zero and thus $\Lambda$ defines a set of measure zero over the domain of $F_{\mathbf{y}_1,\mathbf{y}}(\bm\gamma)$.

\bibliographystyle{IEEEtran}
\bibliography{refdani}

\begin{thebibliography}{10}
\providecommand{\url}[1]{#1}
\csname url@samestyle\endcsname
\providecommand{\newblock}{\relax}
\providecommand{\bibinfo}[2]{#2}
\providecommand{\BIBentrySTDinterwordspacing}{\spaceskip=0pt\relax}
\providecommand{\BIBentryALTinterwordstretchfactor}{4}
\providecommand{\BIBentryALTinterwordspacing}{\spaceskip=\fontdimen2\font plus
\BIBentryALTinterwordstretchfactor\fontdimen3\font minus
  \fontdimen4\font\relax}
\providecommand{\BIBforeignlanguage}[2]{{%
\expandafter\ifx\csname l@#1\endcsname\relax
\typeout{** WARNING: IEEEtran.bst: No hyphenation pattern has been}%
\typeout{** loaded for the language `#1'. Using the pattern for}%
\typeout{** the default language instead.}%
\else
\language=\csname l@#1\endcsname
\fi
#2}}
\providecommand{\BIBdecl}{\relax}
\BIBdecl

\bibitem{patterson1934fourier}
A.~L. Paterson, ``A fourier series method for the determination of the
  components of interatomic distances in crystals,'' \emph{Physical Review},
  vol.~46, no.~5, p. 372, 1934.

\bibitem{patterson1944ambiguities}
A.~L. Patterson, ``Ambiguities in the x-ray analysis of crystal structures,''
  \emph{Physical Review}, vol.~65, no. 5-6, p. 195, 1944.

\bibitem{walther1963question}
A.~Walther, ``The question of phase retrieval in optics,'' \emph{Journal of
  Modern Optics}, vol.~10, no.~1, pp. 41--49, 1963.

\bibitem{millane1990phase}
R.~P. Millane, ``Phase retrieval in crystallography and optics,'' \emph{JOSA
  A}, vol.~7, no.~3, pp. 394--411, 1990.

\bibitem{fienup1987phase}
C.~Fienup and J.~Dainty, ``Phase retrieval and image reconstruction for
  astronomy,'' \emph{Image Recovery: Theory and Application}, pp. 231--275,
  1987.

\bibitem{Stefik1978DNA}
M.~Stefik, ``Inferring dna structures from segmentation data,''
  \emph{Artificial Intelligence}, vol.~11, no.~1, pp. 85--114, 1978.

\bibitem{rabiner1993fundamentals}
L.~Rabiner and B.-H. Juang, ``Fundamentals of speech recognition,'' 1993.

\bibitem{JEB16}
K.~Jaganathan, Y.~C. Eldar, and B.~Hassibi, ``Phase retrieval: An overview of
  recent developments,'' in \emph{Optical Compressive Sensing}, A.~Stern, Ed.,
  2016 (to appear).

\bibitem{shechtman2015phase}
Y.~Shechtman, Y.~C. Eldar, O.~Cohen, H.~N. Chapman, J.~Miao, and M.~Segev,
  ``Phase retrieval with application to optical imaging: a contemporary
  overview,'' \emph{IEEE Signal Processing Magazine}, vol.~32, no.~3, pp.
  87--109, 2015.

\bibitem{hayes1982reconstruction}
M.~H. Hayes, ``The reconstruction of a multidimensional sequence from the phase
  or magnitude of its fourier transform,'' \emph{Acoustics, Speech and Signal
  Processing, IEEE Transactions on}, vol.~30, no.~2, pp. 140--154, 1982.

\bibitem{bates1982fourier}
R.~Bates, ``Fourier phase problems are uniquely solvable in mute than one
  dimension. i: Underlying theory,'' \emph{Optik (Stuttgart)}, vol.~61, pp.
  247--262, 1982.

\bibitem{bruck1979ambiguity}
Y.~M. Bruck and L.~Sodin, ``On the ambiguity of the image reconstruction
  problem,'' \emph{Optics Communications}, vol.~30, no.~3, pp. 304--308, 1979.

\bibitem{H64}
E.~M. Hofstetter, ``Construction of time-limited functions with specified
  autocorrelation functions,'' \emph{Information Theory, IEEE Transactions on},
  vol.~10, no.~2, pp. 119--126, 1964.

\bibitem{HES06}
K.~Huang, Y.~C. Eldar, and N.~Sidiropoulos, ``Phase retrieval from {1D}
  {F}ourier measurements: Convexity, uniqueness, and algorithms,'' 2016,
  submitted to {\em IEEE Trans. on Signal Processing}.

\bibitem{convex}
H.~H. Bauschke, P.~L. Combettes, and D.~R. Luke, ``Phase retrieval, error
  reduction algorithm, and fienup variants: a view from convex optimization,''
  \emph{JOSA A}, vol.~19, no.~7, pp. 1334--1345, 2002.

\bibitem{gerchberg1972practical}
R.~W. Gerchberg, ``A practical algorithm for the determination of phase from
  image and diffraction plane pictures,'' \emph{Optik}, vol.~35, p. 237, 1972.

\bibitem{fienup1978reconstruction}
J.~R. Fienup, ``Reconstruction of an object from the modulus of its fourier
  transform,'' \emph{Optics letters}, vol.~3, no.~1, pp. 27--29, 1978.

\bibitem{streibl1984phase}
N.~Streibl, ``Phase imaging by the transport equation of intensity,''
  \emph{Optics communications}, vol.~49, no.~1, pp. 6--10, 1984.

\bibitem{fienup1982phase}
J.~R. Fienup, ``Phase retrieval algorithms: a comparison,'' \emph{Applied
  optics}, vol.~21, no.~15, pp. 2758--2769, 1982.

\bibitem{shechtman2011sparsity}
Y.~Shechtman, Y.~C. Eldar, A.~Szameit, and M.~Segev, ``Sparsity based
  sub-wavelength imaging with partially incoherent light via quadratic
  compressed sensing,'' \emph{Optics express}, vol.~19, no.~16, pp.
  14\,807--14\,822, 2011.

\bibitem{candes2015phase}
E.~J. Candes, Y.~C. Eldar, T.~Strohmer, and V.~Voroninski, ``Phase retrieval
  via matrix completion,'' \emph{SIAM Review}, vol.~57, no.~2, pp. 225--251,
  2015.

\bibitem{waldspurger2015phase}
I.~Waldspurger, A.~d'Aspremont, and S.~Mallat, ``Phase recovery, {MaxCut} and
  complex semidefinite programming,'' \emph{Mathematical Programming}, vol.
  149, no. 1-2, pp. 47--81, 2015.

\bibitem{C2015}
E.~J. Candes, X.~Li, and M.~Soltanolkotabi, ``Phase retrieval via wirtinger
  flow: Theory and algorithms,'' \emph{Information Theory, IEEE Transactions
  on}, vol.~61, no.~4, pp. 1985--2007, 2015.

\bibitem{BE12}
A.~Beck and Y.~C. Eldar, ``Sparsity constrained nonlinear optimization:
  Optimality conditions and algorithms,'' \emph{SIAM Optimization}, vol.~23,
  no.~3, pp. 1480--1509, Oct. 2013.

\bibitem{shechtman2014gespar}
Y.~Shechtman, A.~Beck, and Y.~C. Eldar, ``{GESPAR}: Efficient phase retrieval
  of sparse signals,'' \emph{Signal Processing, IEEE Transactions on}, vol.~62,
  no.~4, pp. 928--938, 2014.

\bibitem{beinert2015ambiguities}
R.~Beinert and G.~Plonka, ``Ambiguities in one-dimensional discrete phase
  retrieval from fourier magnitudes,'' \emph{Journal of Fourier Analysis and
  Applications}, vol.~21, no.~6, pp. 1169--1198, 2015.

\bibitem{hayes1982red}
M.~H. Hayes and J.~H. McClellan, ``Reducible polynomials in more than one
  variable,'' \emph{Proceedings of the IEEE}, vol.~70, no.~2, pp. 197--198,
  1982.

\bibitem{beinert2016enforcing}
R.~Beinert and G.~Plonka, ``Enforcing uniqueness in one-dimensional phase
  retrieval by additional signal information in time domain,'' \emph{arXiv
  preprint arXiv:1604.04493}, 2016.

\bibitem{hazewinkel2001viete}
M.~Hazewinkel, ``Vi{\`e}te theorem, encyclopedia of mathematics,'' 2001.

\bibitem{gunning2009analytic}
R.~C. Gunning and H.~Rossi, \emph{Analytic functions of several complex
  variables}.\hskip 1em plus 0.5em minus 0.4em\relax American Mathematical
  Soc., 2009, vol. 368.

\bibitem{2015arXiv151207276M}
B.~Mityagin, ``The zero set of a real analytic function,'' \emph{arXiv preprint
  arXiv:1512.07276}, 2015.

\end{thebibliography}

\end{document}